\documentclass[proceedings]{stacs}
\stacsheading{2010}{489-500}{Nancy, France}
\firstpageno{489}

\usepackage{ae,aecompl}
\usepackage{amsthm,amssymb,amsmath,amsfonts}
\usepackage{clrscode}
\usepackage{graphicx,subfig}
\usepackage{amsmath}

\theoremstyle{plain}\newtheorem{fact}[thm]{Fact}

\newcommand{\ADV}{\textsc{Adv}}
\newcommand{\MG}{\textsc{MG}}
\newcommand{\RG}{\textsc{RG}}
\newcommand{\OPT}{\textsc{Opt}}
\newcommand{\ALG}{\textsc{Alg}}
\newcommand{\GALG}{\mathcal{G}_\mathrm{ALG}}
\newcommand{\GRG}{\mathcal{G}_\mathrm{RG}}
\newcommand{\GADV}{\mathcal{G}_\mathrm{ADV}}
\newcommand{\GOPT}{\mathcal{G}_\mathrm{OPT}}
\newcommand{\R}{\mathcal{R}}
\newcommand{\E}{\mathbf{E}}

\begin{document}

\title[Randomized algorithm for agreeable deadlines packet scheduling]%
{A \(\frac{4}{3}\)-competitive randomized algorithm for online scheduling
of packets with agreeable deadlines}

\author[lab1]{{\L}. Je{\.z}}{\L{}ukasz Je\.z}
\address{University of Wroc\l{}aw, Institute of Computer Science}
\email{lje@cs.uni.wroc.pl}
\urladdr{http://www.ii.uni.wroc.pl/~lje/}
\thanks{This work has been supported by MNiSW grants no. N~N206~1723~33, 2007--2010 and N~N206~490638, 2010--2011}


\begin{abstract}
In 2005 Li~et~al. gave a \(\phi\)-competitive deterministic
online algorithm for scheduling of packets with agreeable
deadlines~\cite{DBLP:conf/soda/LiSS05} with a very interesting analysis.
This is known to be optimal due to a lower bound by Hajek~\cite{Hajek-det-lb}.
We claim that the algorithm by Li~et~al. can be slightly simplified, while retaining its
competitive ratio. Then we introduce randomness to the modified algorithm and argue
that the competitive ratio against oblivious adversary is at most \(\frac{4}{3}\).
Note that this still leaves a gap between the best known lower bound of \(\frac{5}{4}\)
by Chin~et~al.~\cite{DBLP:journals/algorithmica/ChinF03} for randomized algorithms against
oblivious adversary.
\end{abstract}

\maketitle

\section{Introduction}

We consider the problem of {\em buffer management with bounded delay}
(aka \emph{packet scheduling}), introduced by
Kesselman et al.~\cite{DBLP:journals/siamcomp/KesselmanLMPSS04}.
It models the behaviour of a~single network switch.
We assume that time is slotted and divided into steps. At the beginning of
a time step, any number of packets may 
arrive at a switch and are stored in its {\em buffer}.
A packet has a positive weight and a deadline, which is the number of step right
before which the packet expires: unless it has already been
transmitted, it is removed from the buffer at the very beginning of that step
and thus can no longer be transmitted.
Only one packet can be transmitted in a single step. 
The goal is to maximize the \emph{weighted throughput}, i.e., the
total weight of transmitted packets.

As the process of managing packet queue is inherently a real-time task,
we investigate the online variant of the problem. This means that the algorithm
has to base its decision of which packet to transmit solely on the 
packets which have already arrived at a switch, without the knowledge
of the future.

%
%

\subsection{Competitive Analysis.}

To measure the performance of an online algorithm, we use a standard notion
of competitive analysis~\cite{online-book}, which compares the gain of
the algorithm to the gain of the optimal solution on the same input sequence.
For any algorithm $\ALG$, we denote its gain on the input sequence $I$ by
$\GALG(I)$; we denote the optimal offline algorithm by $\OPT$. We say that
a deterministic algorithm $\ALG$ is $\R$-competitive if on any input
sequence~$I$,
it holds that $\GALG(I) \geq \frac 1 \R \cdot \GOPT(I)$. 

When analysing the performance of an online algorithm $\ALG$, we view the process
as a game between $\ALG$ and an {\em adversary}. The adversary controls
the packets' injection into the buffer and chooses which of them to send.
The goal is then to show that the adversary's gain is at most $\R$ times
$\ALG$'s gain.

If the algorithm is randomized, we consider its expected gain, $\E[\GALG(I)]$, where the 
expectation is taken over all possible random choices made by $\ALG$. 
However, in the randomized case, the power of the adversary has to be further specified.
Following Ben-David et al.~\cite{DBLP:journals/algorithmica/Ben-DavidBKTW94},
we distinguish between an {\em oblivious} and {\em adaptive-online} adversary
(called adaptive for short). An oblivious adversary has to construct the whole input sequence
in advance,  not knowing the random bits used by an algorithm. The expected gain
of $\ALG$ is compared to the gain of the optimal offline solution on $I$. 
An adaptive adversary decides packet injections
upon seeing which packets are transmitted by the algorithm.
However, it has to provide an answering entity~$\ADV$, which creates
a solution on-line (in parallel to $\ALG$) and cannot change it afterwards.
We say that $\ALG$ is $\R$-competitive against an adaptive adversary if for any input sequence $I$
created adaptively and any answering algorithm $\ADV$, it holds that 
$\E[\GALG(I)] \geq \frac{1}{\R} \cdot \E[\GADV(I)]$. We note that $\ADV$ is (wlog) deterministic, 
but as $\ALG$ is randomized, so is the input sequence $I$.

In the literature on online algorithms (see e.g.~\cite{online-book}), 
the definition of the competitive ratio sometimes allows an additive constant,
i.e., a~deterministic algorithm is $\R$-competitive if there exists a constant $\alpha \geq 0$
such that $\GALG(I) \geq \frac 1 \R \cdot \GOPT(I) - \alpha$ holds for evry input sequence~$I$.
An analogous definition applies to randomized case. Our upper bounds hold for $\alpha = 0$.

\subsection{Previous work}

The best known deterministic and randomized algorithms for general
instances have competitive ratios at most
\(2\sqrt{2}-1\approx 1.828\)~\cite{DBLP:conf/soda/EnglertW07} and
\(e/(e-1)\approx 1.582\)~\cite{DBLP:journals/jda/ChinCFJST06},
respectively. A recent analysis of the latter algorithm shows that it
retains its competitive ratio even against adaptive-online
adversary~\cite{DBLP:journals/corr/abs-0907-2050}.

The best known lower bounds on competitive ratio against either adversary type
use rather restricted \emph{\(2\)-bounded} sequences in which every packet has
lifespan (deadline~$-$~release time) either \(1\) or \(2\). The lower bounds in
question are \(\phi\approx 1.618\) for deterministic algorithms~\cite{Hajek-det-lb},
\(\frac{4}{3}\) for randomized algorithms against adaptive
adversary~\cite{DBLP:conf/waoa/BienkowskiCJ08},
and \(\frac{5}{4}\) for randomized algorithms against oblivious
adversary~\cite{DBLP:journals/algorithmica/ChinF03}.
All these bounds are tight for \(2\)-bounded
sequences~\cite{DBLP:journals/siamcomp/KesselmanLMPSS04,%
DBLP:conf/waoa/BienkowskiCJ08,DBLP:journals/jda/ChinCFJST06}.

We restrict ourselves to sequences with \emph{agreeable deadlines}, in which
packets released later have deadlines at least as large as those released before
(\(r_i < r_j\) implies \(d_i \leq d_j\)). These strictly generalize the
\(2\)-bounded sequences. Sequences with agreeable deadlines also properly contain
\emph{s-uniform} sequences for all \(s\), i.e., sequences in which every packet has
lifespan exactly \(s\). An optimal \(\phi\)-competitive deterministic
algorithm for sequences with agreeable deadlines is known~\cite{DBLP:conf/soda/LiSS05}.

Je\.zabek studied the impact of \emph{resource augmentation} on the deterministic competitive
ratio~\cite{DBLP:conf/sofsem/Jezabek09,Jezabek-Agreeable}. It turns out that while
allowing the deterministic algorithm to transmit \(k\) packets in a single step for
any constant \(k\) cannot make it \(1\)-competitive (compared to the single-speed
offline optimum) on unrestricted sequences~\cite{DBLP:conf/sofsem/Jezabek09},
\(k=2\) is sufficient for sequences with agreeable deadlines~\cite{Jezabek-Agreeable}.

\subsection{Our contribution}

Motivated by aforementioned results for sequences with agreeable deadlines,
we investigate randomized algorithms for such instances. We devise a
\(\frac{4}{3}\)-competitive randomized algorithm against oblivious adversary.
The algorithm and its analysis are inspired by those by
Li~et~al.~\cite{DBLP:conf/soda/LiSS05} for deterministic case.
The key insight is as follows. The algorithm {\MG} by Li~et~al.~\cite{DBLP:conf/soda/LiSS05}
can be simplified by making it always send either \(e\), the heaviest among the
earliest non-dominated packets, or \(h\), the earliest among the heaviest non-dominated packets.
We call this algorithm \(\MG'\), and prove that it remains \(\phi\)-competitive.
Then we turn it into a randomized algorithm {\RG}, simply by making it always transmit \(e\)
with probability \(\frac{w_e}{w_h}\) and \(h\) with the remaining probability. The proof
of {\RG}'s \(\frac{4}{3}\)-competitiveness against oblivious adversary follows by
similar analysis.

\section{Preliminaries}

We denote the release time, weight, and deadline of a packet \(j\) by \(r_j\),
\(w_j\), and \(d_j\), respectively. A packet \(j\) is \emph{pending} at step
\(t\) if \(r_j \leq t\), it has not yet been transmitted, and \(d_j>t\). We
introduce a linear order \(\unlhd\) on the packets as follows: \(i\unlhd j\)
if either
\begin{align*}
	&d_i<d_j\enspace\text{, or}\\
	&d_i=d_j\enspace\text{and }w_i>w_j\enspace\text{, or}\\
	&d_i=d_j\enspace\text{and }w_i=w_j\enspace\text{and }r_i \leq r_j\enspace.
\end{align*}
To make \(\unlhd\) truly linear we assume
that in every single step the packets are released one after another rather
then all at once, e.g. that they have unique fractional release times.

A \emph{schedule} is a mapping from time steps to packets to be transmitted in
those time steps. A schedule is \emph{feasible} if it is injective and for every
time step \(t\) the packet that \(t\) maps to is pending at~\(t\).
It is convenient to view a feasible schedule \(S\)
differently, for example as the set \(\{S(t)\colon t>0\}\),
the sequence \(S(1),S(2),\ldots\),
or a matching in the \emph{schedulability graph}. The schedulability graph is
a bipartite graph, one of whose partition classes 
is the set of packets and the other is the set of time steps. Each packet \(j\) is
connected precisely to each of the time steps \(t\) such that \(r_j \leq t < d_j\)
by an edge of weight \(w_j\); an example is given in Figure~\ref{fig:sch-graph}.
Observe that optimal offline schedules correspond to maximum weight matchings
in the schedulability graph. Thus an optimal offline schedule can be found in
polynomial time using the classic ``Hungarian algorithm'', see for
example~\cite{comb-opt-book}.
One may have to remove appropriately chosen time step vertices first, so that the
remaining ones match the number of packet vertices, though.

\begin{figure}[h]
	\begin{center}
		\includegraphics[width=0.64\textwidth]{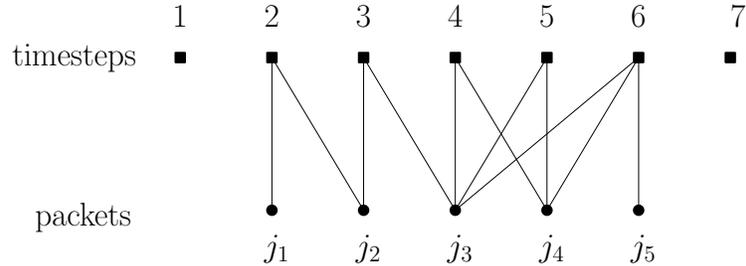}
	\end{center}
	\caption{Schedulability graph for packets \(j_1,j_2,\ldots,j_5\), whose release times
		and deadlines are (2,3), (2,4), (3,7), (4,7), (6,7) respectively;
		we ignore packet weights in the figure.
		Packets are represented by discs, time steps by squares.}
    \label{fig:sch-graph}
\end{figure}

Given any linear order \(\preceq\) on packets and a (feasible) schedule \(S\),
 we say that \(S\) is consistent with \(\preceq\), or that \(S\) is a \(\preceq\)-schedule,
 if for every \(t\) the packet \(S(t)\) is the minimum pending packet with respect
to \(\preceq\). It is fairly easy to observe that if \(\preceq\) is any
\emph{earliest deadline first}, with ties broken in an arbitrary way,
then any feasible schedule can be turned to a unique \(\preceq\)-schedule by reordering
its packets; in particular this applies to \(\unlhd\).

Recall that the \emph{oblivious} adversary prepares the whole input sequence in
advance and cannot alter it later on. Thus its solution is simply the offline
optimal schedule for the complete sequence. Nevertheless, we still refer to
the answering entity {\ADV} rather than {\OPT} in our analysis, as it involves
altering the set of packets pending for the adversary, which may well be viewed
as altering the input sequence. Now we introduce two schedules that
are crucial for our algorithms and our analyzes.
\begin{definition}
The \emph{oblivious schedule} at time step \(t\), denoted \(O_t\), is any
fixed optimal feasible \(\unlhd\)-schedule over all the packets pending
at step \(t\). For fixed~\(O_t\), a~packet~\(j\) pending at~\(t\) is called
\emph{dominated} if \(j \notin O_t\), and \emph{non-dominated} otherwise. For
fixed~\(O_t\) let \(e\) denote \(O_t(t)\), the \(\unlhd\)-minimal of all non-dominated
packets, and \(h\) denote the \(\unlhd\)-minimal of all non-dominated maximum-weight
packets.
\end{definition}
Note that both the adversary and the algorithm can calculate their oblivious
schedules at any step, and that these will coincide if their buffers are the same.
\begin{definition}
For a fixed input sequence, the \emph{clairvoyant schedule} at time step \(t\),
denoted \(C_t\), is any fixed optimal feasible schedule over all the packets
pending at step \(t\) and all the packets that will arrive in the future.
\end{definition}
Naturally, the adversary can calculate the clairvoyant schedule, as it knows
the fixed input sequence, while the algorithm cannot, since it only knows the
part of input revealed so far. However, the oblivious schedule gives some partial
information about the clairvoyant schedule: intuitively, if \(p\) is dominated at
\(t\), it makes no sense to transmit it at \(t\). Formally, (wlog) dominated
packets are not included in the clairvoyant schedule, as stated in the following.
\begin{fact}\label{fact:O-contains-C}
For any fixed input sequence, time step \(t\), and oblivious schedule~\(O_t\),
there is a clairvoyant schedule~\(C^*_t\) such that
\(C^*_t\cap\{j\colon r_j \leq t\} \subseteq O_t\).
\end{fact}
\begin{proof}
This is a standard alternating path argument about matchings. If you are unfamiliar
with these concepts, refer to a book by A.~Schrijver~\cite{comb-opt-book} for example.

Let \(O_t\) be the oblivious schedule and \(C_t\) be any clairvoyant schedule.
Treat both as matchings in the schedulability graph and consider their
symmetric difference \(C_t \oplus O_t\). Consider any job \(j \in C_t \setminus O_t\)
such that \(r_j \leq t\). It is an endpoint of an alternating path \(P\) in
\(C_t \oplus O_t\). Note that all the jobs on \(P\) are already pending at time \(t\):
this is certainly true about \(j\), and all the successive jobs belong to \(O_t\), so
they are pending as well.

First we prove that \(P\) has even length, i.e., it ends in a node corresponding to a job.
Assume for contradiction that \(P\)'s length is odd, and that \(P\) ends in a node
corresponding to a timestep \(t'\). Note that no job is assigned to \(t'\) in~\(O_t\).
Then \(O_t \oplus P\) is a feasible schedule that, treated as a set, satisfies
\(O_t \subseteq O_t \oplus P\) and \(j \in O_t \oplus P\). This contradicts optimality
of~\(O_t\). See Figure~\ref{fig:odd-path} for illustration.

Thus \(P\) has even length and ends with a job \(j' \in O_t \setminus C_t\).
By optimality of both \(O_t\) and \(C_t\), \(w_j=w_{j'}\) holds.
Thus \(C_t \oplus P\) is an optimal feasible schedule: in terms of sets the only
difference between \(C_t\) and \(C_t \oplus P\) is that \(j\) has been replaced
by \(j'\), a job of the same weight. See Figure~\ref{fig:even-path} for illustration.

\begin{figure}[h]%
	\centering
	\subfloat[%
		\(P\) cannot have odd length: in such case the assignment of jobs on \(P\)
		in \(O_t\) could be changed to match the strictly better assignment of \(C_t\).%
	]{%
		\label{fig:odd-path}\includegraphics[width=0.41\textwidth]{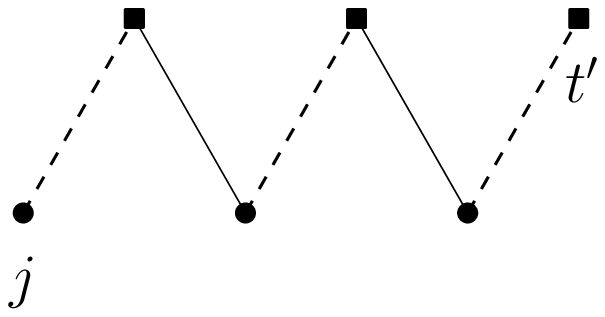}
	}
	\qquad
	\subfloat[%
		\(P\) has even length: now the assignment of jobs on \(P\) in \(C_t\) can be
		changed to match the assignment of \(O_t\) so that the value of \(\Delta\) drops.%
	]{%
		\label{fig:even-path}\includegraphics[width=0.495\textwidth]{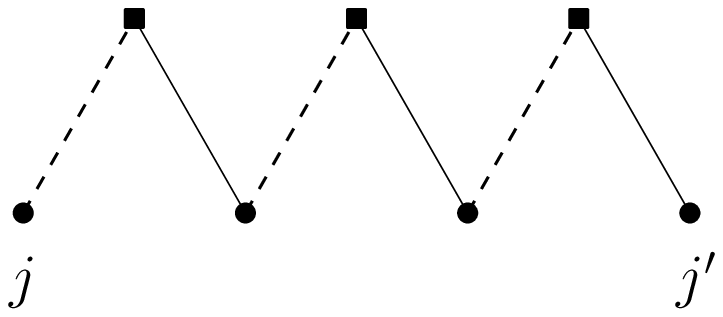}
	}
  \caption{The alternating path \(P\). Packets are represented by discs, time steps by squares.
		Dashed lines represent \(C_t\), solid lines represent \(O_t\).}
	%
\end{figure}

Applying such changes iteratively transforms \(C_t\) to
a clairvoyant schedule~\(C^*_t\) as announced.
To observe that a finite number of iterations suffices, define
\(\Delta(S):=|S\cap\{j\colon r_j \leq t\} \setminus O_t|\) for any schedule \(S\).
It follows that \(\Delta(C_t \oplus P) = \Delta(C_t)-1\).
Since \(\Delta\) is non-negative and its value drops by one with each iteration,
\(C^*_t\) is obtained in a finite number of steps.
\end{proof}
\begin{definition}\label{dfn:conforming-schedules}
We say that a clairvoyant schedule \(C_t\) conforms with an oblivious schedule~\(O_t\)
if \(C_t\) is a~\(\unlhd\)-schedule, \(C_t\cap\{j\colon r_j \leq t\} \subseteq O_t\),
and for all \(i \in O_t\) such that \(i \lhd j=C_t(t)\), \(w_i < w_j\) holds.
\end{definition}
\begin{fact}\label{fact:adv-takes-not-too-late}
For every oblivious schedule \(O_t\) there is a conforming clairvoyant schedule \(C^*_t\).
\end{fact}
\begin{proof}
Let \(C_t\) be a clairvoyant schedule such that \(C_t\cap\{j\colon r_j \leq t\} \subseteq O_t\);
Fact~\ref{fact:O-contains-C} guarantees its existence. Let \(C^*_t\) be the schedule
obtained from \(C_t\) by first turning it into a \(\unlhd\)-schedule \(C'_t\) and then
replacing \(j=C'_t(t)\) with a \(\unlhd\)-minimal non-dominated packet \(j'\)
of the same weight.

If \(j'=j\), then \(C^*_t=C'_t\), and thus it is a clairvoyant \(\unlhd\)-schedule.
Assume \(j' \neq j\), i.e., \(j' \lhd j\). Then \(j' \notin C'_t\), since \(C'_t\) is a \(\unlhd\)-schedule.
Thus \(C^*_t\) is feasible as we replace \(C'_t\)'s very first packet by another pending packet
which was not included in \(C_t\). Observe that \(C^*_t\) is indeed a clairvoyant \(\unlhd\)-schedule:
optimality follows from \(w_{j'}=w_j\), while consistency with \(\unlhd\) follows from \(j' \lhd j\).

It remains to prove that for every \(i \in O_t\) such that \(i \lhd j'\), \(w_i < w_{j'}=w_j\) holds.
Note that \(i \notin C^*_t\) as \(C^*_t\) is a \(\unlhd\)-schedule, and that \(w_i \neq w_{j'} = w_j\)
holds by the choice of \(j'\). Assume for contradiction that \(w_i > w_j\). Then \(C^*_t\) with \(j\)
replaced by \(i\) is a feasible schedule contradicting optimality of \(C^*_t\).
\end{proof}
Now we inspect some properties of conforming schedules.
\begin{fact}\label{fact:adv-takes-sth-earlier}
Let \(C_t\) be a clairvoyant schedule conforming with an oblivious schedule~\(O_t\).
If \(i,j \in O_t\), \(w_i<w_j\) and \(d_i<d_j\) (or, equivalently\, \(w_i<w_j\) and \(i\lhd j\)),
and \(i \in C_t\), then also \(j \in C_t\).
\end{fact}
\begin{proof}
Assume for contradiction that \(j \notin C_t\). Then \(C_t\) with \(i\) replaced by \(j\)
is a feasible schedule contradicting optimality of \(C_t\).
\end{proof}
\begin{lemma}\label{lem:reordering}
Let \(C_t\) be a clairvoyant schedule conforming with an oblivious schedule~\(O_t\).
Suppose that \(e=O_t(t) \notin C_t\). Then there is a clairvoyant schedule \(C^*_t\)
obtained from \(C_t\) by reordering of packets such that \(C^*_t(t)=h\).
\end{lemma}
\begin{proof}
Let \(j=C_t(t) \neq h\) and let \(O_t=p_1,p_2,\ldots,p_s\). Observe that
\(h \in C_t\) by Fact~\ref{fact:adv-takes-sth-earlier}.
So in particular \(e=p_1\), \(j=p_k\), and \(h=p_l\) for some \(1<k<l\leq s\).
Let \(d_i\) denote the deadline of \(p_i\) for \(1\leq i \leq s\).
Since \(O_t\) is feasible in the absence of future arrivals,
\(d_i \geq t+i\) for \(i=1,\ldots,s\).

Recall that \(p_k,p_l\in C_t\) and that there can be some further packets
\(p\in C_t\) such that \(p_k \lhd p \lhd p_l\); some of these packets may
be not pending yet. We construct a schedule \(C'_t\) by reordering \(C_t\).
Precisely, we put all the packets from \(C_t\) that are not yet pending at
\(t\) after all the packets from \(C_t\) that are already pending, keeping
the order between the pending packets and between those not yet pending.
By the agreeable deadlines property, this is an earliest deadline first order,
so \(C'_t\) is a clairvoyant schedule.

As \(e=p_1 \notin C'_t\) and \(d_i \geq t+i\) for \(i=1,\ldots,s\), all the packets
\(x \in C'_t\) preceding \(h\) in \(C'_t\) (i.e., \(x \in C'_t\) such that \(r_x\leq t\)
and \(x \lhd p_l=h\)) have \emph{slack} in \(C'_t\), i.e., each of them could also be
scheduled one step later. Hence \(h=p_l\) can be moved to the very front of \(C'_t\)
while keeping its feasibility, i.e., \(C'_t=p_k,p_{k'},\ldots,p_{l'},p_l\) can be
transformed to a clairvoyant schedule \(C^*_t=p_l,p_k,p_{k'},\ldots,p_{l'}\).
The reordering is illustrated in Figure~\ref{fig:reorder}.
\begin{figure}[h]
  \centering
    \includegraphics[width=0.85\textwidth]{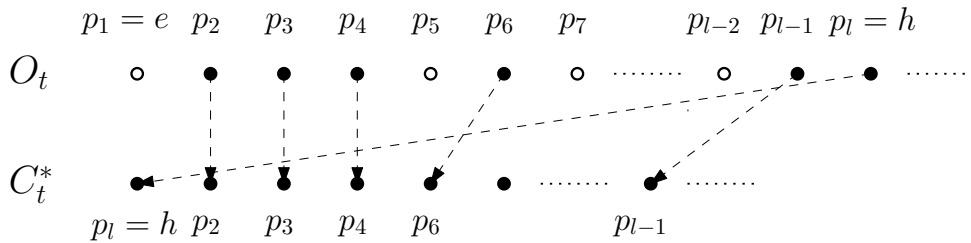}
    \caption{Construction of the schedule \(C^*_t\). Packets are represented by circles:
    	the ones included in \(C_t\) (\(C^*_t\)) are filled, the remaining ones are hollow.}
    \label{fig:reorder}
\end{figure}
\end{proof}


\section{Algorithms and their analyzes}

\subsection{The Algorithms}

The algorithm {\MG}~\cite{DBLP:conf/soda/LiSS05} works as follows: at the
beginning of each step \(t\) it considers the packets in the buffer and the
newly arrived packets, and calculates~\(O_t\). Then {\MG} identifies the
packets \(e\) and \(h\). If \(\phi w_e \geq w_h\), {\MG} sends \(e\).
Otherwise, it sends the \(\unlhd\)-minimal packet \(f\) such that
\(w_f \geq \phi w_e\) and \(\phi w_f \geq w_h\);
the latter exists as \(h\) itself is a valid candidate.
Our deterministic algorithm {\MG\('\)} does exactly the same with one exception:
if \(\phi w_e < w_h\), it sends \(h\) rather than \(f\). Our randomized algorithm
{\RG} also works in a similar fashion: it transmits \(e\) with probability \(\frac{w_e}{w_h}\)
and \(h\) with the remaining probability. For completeness, we provide pseudo-codes
of all three algorithms in Figure~\ref{fig:codes}.

\begin{figure}[h]
	\centering
	\begin{codebox}
	\Procname{\proc{MG} (step \(t\))}
	\zi	\(O_t \gets\) oblivious schedule at \(t\)
	\zi	\(e \gets\) the \(\unlhd\)-minimal packet from \(O_t\)
	\zi	\(h \gets \) the \(\unlhd\)-minimal of all the heaviest packets from \(O_t\)
	\zi	\If \(\phi w_e \geq w_h\)
	\zi		\Then
					transmit \(e\)
	\zi		\Else
					\(f \gets\) the \(\unlhd\)-minimal of all \(j \in O_t\) s.t. \(w_j \geq \phi w_e\) and \(\phi w_j \geq w_h\)
	\zi			transmit \(f\)
			\End
	\end{codebox}
	\begin{codebox}
	\Procname{\proc{MG}\('\) (step \(t\))}
	\zi	\(O_t \gets\) oblivious schedule at \(t\)
	\zi	\(e \gets\) the \(\unlhd\)-minimal packet from \(O_t\)
	\zi	\(h \gets \) the \(\unlhd\)-minimal of all the heaviest packets from \(O_t\)
	\zi	\If \(\phi w_e \geq w_h\)
	\zi		\Then
					transmit \(e\)
	\zi		\Else
					transmit \(h\)
			\End
	\end{codebox}
	\begin{codebox}
	\Procname{\proc{RG} (step \(t\))}
	\zi	\(O_t \gets\) oblivious schedule at \(t\)
	\zi	\(e \gets\) the \(\unlhd\)-minimal packet from \(O_t\)
	\zi	\(h \gets \) the \(\unlhd\)-minimal of all the heaviest packets from \(O_t\)
	\zi	transmit \(e\) with probability \(\frac{w_e}{w_h}\) and \(h\) with probability \(1-\frac{w_e}{w_h}\)
	\end{codebox}
	\caption{The three algorithms}
	\label{fig:codes}
\end{figure}

\subsection{Analysis Idea}

The analysis of Li~et~al.~\cite{DBLP:conf/soda/LiSS05} uses the following idea:
in each step, after both {\MG} and {\ADV} transmitted their packets,
modify {\ADV}'s buffer in such a way that it remains the same as
{\MG}'s and that this change can only improve {\ADV}'s gain, both in this
step and in the future. Sometimes {\ADV}'s schedule is also modified to
achieve this goal, specifically, the packets in it may be reordered, and {\ADV}
may sometimes be allowed to transmit two packets in a single step.
It is proved that in each such step the ratio of {\ADV}'s to {\MG}'s gain is
at most \(\phi\). As was already noticed by Li~et~al.~\cite{DBLP:conf/soda/LiSS05},
this is essentially a potential function argument. To simplify the analysis, it is
assumed (wlog) that {\ADV} transmits its packets in the \(\unlhd\) order.

Our analysis follows the outline of the one by Li~et~al., but we make it more
formal. Observe that there may be multiple clairvoyant schedules, and that
{\ADV} can transmit \(C_t(t)\) at every step \(t\), where \(C_t\) is a clairvoyant
schedule chosen arbitrarily in step \(t\). As our algorithms \(\MG'\) and {\RG}
determine the oblivious schedule \(O_t\) at each step, we assume that every \(C_t\)
is a clairvoyant schedule conforming with \(O_t\).

There is one exception though. Sometimes, when a reordering in \(C_t\) does not hinder
{\ADV}'s performance (taking future arrivals into account), we ``force'' {\ADV} to
follow the reordered schedule. This is the situation described in
Lemma~\ref{lem:reordering}: when \(e \notin C_t\), there is a clairvoyant
schedule \(C^*_t\) such that \(h=C^*_t(t)\). In such case we may assume that {\ADV}
follows \(C^*_t\) rather than \(C_t\), i.e., that it transmits \(h\) at \(t\). Indeed,
we make that assumption whenever our algorithm (either {\MG}\('\) or {\RG}) transmits
\(h\) at such step: then {\ADV} and {\MG}\('\) (\RG) transmit the same packet,
which greatly simplifies the analysis.

Our analysis of \(\MG'\) is essentially the same as the original analysis of {\MG} by
Li~et~al.~\cite{DBLP:conf/soda/LiSS05}, but lacks one case which is superfluous due to
our modification. As our algorithm \(\MG'\) always transmits either \(e\) or \(h\),
and the packet \(j\) that {\ADV} transmits always satisfies \(j \unlhd h\) by
definition of the clairvoyant schedule	conforming with \(O_t\), the case which {\MG}
transmits \(f\) such that \(e \lhd f \lhd j\) does not occur to \(\MG'\).
The same observation applies to {\RG}, whose analysis also follows the ideas of Li~et~al.

\subsection{Analysis of the Deterministic Algorithm}

We analyze this algorithm as mentioned before, i.e., assuming (wlog) that at
every step \(t\) {\ADV} transmits \(C_t(t)\), where \(C_t\) is a clairvoyant
schedule conforming with \(O_t\).
%
\begin{theorem}\label{th: semi-old-result}
{\MG\('\)} is \(\phi\)-competitive on sequences with agreeable deadlines.
\end{theorem}
\begin{proof}
Note that whenever \(MG'\) and {\ADV} transmit the same packet, clearly their gains
are the same, as are their buffers right after such step. In particular this happens
when \(e=h\) as then \(\MG'\) transmits \(e=h\) and {\ADV} does the same:
in such case \(h\) is both the heaviest packet and the \(\unlhd\)-minimal
non-dominated packet, so \(h = C_t(t)\) by definition of the clairvoyant schedule
conforming with \(O_t\).

 In what follows we inspect the three remaining cases.
\paragraph{\boldmath\(\phi w_e \geq w_h\colon \MG' \text{ transmits } e. \ \ADV \text{ transmits } j \neq e.\)}
	To make the buffers of {\MG\('\)} and {\ADV} identical right after this step,
	we replace \(e\) in {\ADV}'s buffer by~\(j\).
	This is advantageous for {\ADV} as \(d_j \geq d_e\) and \(w_j \geq w_e\) follows from \(e \unlhd j\)
	and the definition of a clairvoyant schedule conforming with \(O_t\).
	As \(\phi w_e \geq w_h\), the ratio of gains is
	\[
	  \frac{w_j}{w_e} \leq \frac{w_h}{w_e} \leq \phi \enspace.
	\]
\paragraph{\boldmath\(\phi w_e < w_h\colon \MG' \text{ transmits } h. \ \ADV \text{ transmits } e.\)}
	Note that {\ADV}'s clairvoyant schedule from this step contains \(h\) by
	Fact~\ref{fact:adv-takes-sth-earlier}.
	We let {\ADV} transmit 	both \(e\) and \(h\) in this step and keep \(e\) in its buffer,
	making it identical to the buffer of {\MG\('\)}. Keeping \(e\), as well as transmitting
	two packets at a time is clearly advantageous for {\ADV}. As \(\phi w_e < w_h\), the ratio
	of gains is
	\[
		\frac{w_e+w_h}{w_h} \leq \frac{1}{\phi} + 1 = \phi \enspace.
	\]
\paragraph{\boldmath\(\phi w_e < w_h\colon \MG' \text{ transmits } h. \ \ADV \text{ transmits } j \neq e.\)}
	Note that \(j \unlhd h\): by definition of the clairvoyant schedule
	conforming with \(O_t\), for every \(i \in O_t\) such that
	\(i \lhd j\), \(w_i < w_j\) holds.
	
	There are two cases: either \(j=h\), or \(w_j < w_h\) and \(d_j < d_h\). In the former one both
	players do the same and end up with identical buffers. Thus we focus on the latter case.
	Fact~\ref{fact:adv-takes-sth-earlier} implies that \(h \in C_t\). By Lemma~\ref{lem:reordering},
	\(C_t\) remains feasible when \(h\) is moved to its very beginning. Hence we assume that {\ADV}
	transmits \(h\) in the current step. As this is the packet that {\MG\('\)} sends, the gains of
	{\ADV} and \(\MG'\) are the same and no changes need be made to {\ADV}'s buffer.
\end{proof}

\subsection{Analysis of the Randomized Algorithm}

We analyze this algorithm as mentioned before, i.e., assuming (wlog) that at
every step \(t\) {\ADV} transmits \(C_t(t)\), where \(C_t\) is a clairvoyant
schedule conforming with \(O_t\).
%
\begin{theorem}\label{th: result}
{\RG} is \(\frac{4}{3}\)-competitive against oblivious adversary on sequences
with agreeable deadlines.
\end{theorem}
\begin{proof}
Observe that if \(e=h\), then {\RG} transmits \(e=h\) and {\ADV} does the same:
as in such case \(h\) is both the heaviest packet and the \(\unlhd\)-minimal
non-dominated packet, \(h = C_t(t)\) by definition of the clairvoyant schedule
conforming with \(O_t\). In such case the gains of {\RG} and {\ADV} 
are clearly the same, as are their buffers right after step \(t\).
Thus we assume \(e \neq h\) from now on.

Let us first bound the algorithm's expected gain in one step. It equals
\begin{align}\label{eq: galg}
\GRG 	&= \frac{w_e}{w_h} \cdot w_e + \left(1-\frac{w_e}{w_h}\right)\cdot w_h \nonumber\\
		&= \frac{1}{w_h}\left(w_e^2-w_ew_h+w_h^2\right) \nonumber\\
		&= \frac{1}{w_h}\left(\left(w_e-\frac{w_h}{2}\right)^2+\frac{3}{4}w_h^2\right) \nonumber\\
		&\geq \frac{3}{4}w_h \enspace.
\end{align}

Now we describe the changes to {\ADV}'s scheduling policy and buffer
in the given step. These make {\ADV}'s {\RG}'s buffers identical,
and, furthermore, make the expected gain of the adversary equal
exactly $w_h$. This, together with~\eqref{eq: galg} yields the
desired bound. To this end we consider cases depending on {\ADV}'s choice.
\begin{enumerate}
\item {\ADV} transmits \(e\).
	Note that {\ADV}'s clairvoyant schedule from this step
	contains \(h\) by Fact~\ref{fact:adv-takes-sth-earlier}.
	
	If {\RG} transmits \(e\), which it does with probability
	\(\frac{w_e}{w_h}\), both players gain \(w_e\) and no
	changes are required.
	
	Otherwise {\RG} transmits \(h\), and we let {\ADV} transmit
	both \(e\) and \(h\) in this step and keep \(e\) in its buffer,
	making it identical to {\RG}'s buffer. Keeping \(e\), as well
	as transmitting two packets at a time is clearly advantageous
	for {\ADV}.
	
	Thus in this case the adversary's expected gain is
	\[
	\GADV
		= \frac{w_e}{w_h} \cdot w_e +
			\left(1-\frac{w_e}{w_h}\right)\left(w_e+w_h\right)
		= w_e + (w_h - w_e)
		= w_h \enspace .
	\]
\item {\ADV} transmits \(j \neq e\).
	Note that \(j \unlhd h\): by definition of the clairvoyant schedule
	conforming with \(O_t\), for every \(i \in O_t\) such that
	\(i \lhd j\), \(w_i < w_j\) holds.
	
	If {\RG} sends \(e\), which it does with probability
	\(\frac{w_e}{w_h}\), we simply replace \(e\) in {\ADV}'s buffer by
	\(j\). This is advantageous for {\ADV} as \(w_j > w_e\) and \(d_j > d_e\)
	follow from \(e \lhd j\) and the definition of the clairvoyant schedule
	conforming with \(O_t\).
	
	Otherwise {\RG} sends \(h\), and we claim that (wlog) {\ADV} does the same.
	Suppose that \(j \neq h\), which implies that \(w_j < w_h\) and \(d_j < d_h\).
	Then \(h \in C_t\), by Fact~\ref{fact:adv-takes-sth-earlier}.
	Thus, by Lemma~\ref{lem:reordering},
	\(C_t\) remains feasible when \(h\) is moved to its very beginning.
	Hence we assume that {\ADV} transmits \(h\) in the current step.
	No further changes need be made to {\ADV}'s buffer as {\RG} also
	sends \(h\).
	
	Thus in this case the adversary's expected gain is $w_h$.\qedhere
\end{enumerate}
\end{proof}

\section{Conclusion and Open Problems}

We have shown that, as long as the adversary is oblivious, the ideas
of Li~et~al.~\cite{DBLP:conf/soda/LiSS05} can be applied to randomized algorithms,
and devised a \(\frac{4}{3}\)-competitive algorithm this way. However,
the gap between the \(\frac{5}{4}\) lower bound and our \(\frac{4}{3}\)
upper bound remains.

Some parts of our analysis hold even in the adaptive
adversary model~\cite{DBLP:journals/corr/abs-0907-2050}.
On the other hand, other parts do not extend to adaptive adversary
model, since in general such adversary's schedule is a random variable
depending on the algorithm's random choices. Therefore it is not possible to
assume that this ``schedule'' is ordered by deadlines,
let alone perform reordering like the one in proof of Lemma~\ref{lem:reordering}.

This makes bridging either the \(\left[\frac{5}{4},\frac{4}{3}\right]\) gap in the
oblivious adversary model, or the \(\left[\frac{4}{3},\frac{e}{e-1}\right]\) gap
in the adaptive adversary model all the more interesting.

\subsection*{Acknowledgements}

I would like to thank my brother, Artur Je\.z, for numerous comments on the draft of this paper.

\bibliographystyle{abbrv} 
\bibliography{jez-lukasz}

\newpage
\strut

\end{document}